%% file: draft.tex
\documentclass[article]{llncs}

\usepackage[utf8]{inputenc}
\usepackage[T1]{fontenc}
\usepackage{complexity} 
\usepackage[numbers]{natbib}
\usepackage{hyperref}

\input macrolocals

\title{On the complexity of Linearizability}

\begin{document}

\author{Jad Hamza}
\institute{LIAFA, Université Paris Diderot}

\maketitle

\input{abstract}

\input{introduction}
\input{definitions}
\input{linearizability}
\input{proof}
\input{conclusion}

\bibliographystyle{splncsnat}
\bibliography{violin}

\end{document}

%% file: macrolocals.tex
\usepackage{mathtools}
\usepackage{pgf}
\usepackage{tikz}
\usepackage{amssymb}
\usepackage{amsmath}

\usepackage{amsthm}

\usetikzlibrary{calc}
\usetikzlibrary{arrows}
\usetikzlibrary{automata}

\newcommand{\wrt}{w.r.t.~}

\newcommand{\fig}[1]{Fig~\ref{#1}}

\newcommand{\lem}[1]{Lemma~\ref{#1}}

\newcommand{\set}[1]{\{{#1}\}}
\newcommand{\sect}[1]{Section~\ref{#1}}
\newcommand{\ab}{\allowbreak}

\newtheorem*{lemma*}{Lemma}

\newcommand{\tmachine}{Turing machine}
\newcommand{\tm}{\mathcal{M}}
\newcommand{\States}{Q}
\newcommand{\transitions}{\delta}

\newcommand{\istate}{q_0}
\newcommand{\fstate}{q_f}
\newcommand{\state}{q}
\newcommand{\cell}{c}
\newcommand{\lefttr}{\leftarrow}
\newcommand{\righttr}{\rightarrow}
\newcommand{\tminput}{t}
\newcommand{\tminputsize}{n}

\newcommand{\thelog}{{\polynomial(\tminputsize)}}
\newcommand{\ncells}{2^\thelog}

\newcommand{\Linearizability}{Linearizability}

\newcommand{\problemx}{Letter Insertion}
\newcommand{\letter}{a}
\newcommand{\nletter}{l}
\newcommand{\Letters}{A}
\newcommand{\insertable}{insertable}

\newcommand{\aut}{N}
\newcommand{\word}{w}
\newcommand{\perm}{p}

\newcommand{\minisep}{:}
\newcommand{\sepcell}{;}
\newcommand{\callconfig}{\$}
\newcommand{\retconfig}{\hookleftarrow}
\newcommand{\callrun}{\triangleright}
\newcommand{\retrun}{\square}
\newcommand{\theint}[1]{{\bf #1}}
 
\newcommand{\notwellformed}{\aut_{\sf notWF}}
\newcommand{\seqconf}{sequence of configurations}
\newcommand{\wellformedcfg}{well-formed configuration}
\newcommand{\wellformedseq}{well-formed sequence of configurations}

\newcommand{\notrun}{\aut_{\sf NotSeqCfg}}
\newcommand{\notdelta}{\aut_{\sf NotDelta}}
\newcommand{\notdeltaexample}[1]{\aut_{#1}} 
\newcommand{\notdeltaex}{\notdeltaexample{((\state,0),(\state',1,\righttr))}^{(1)}}
\newcommand{\pletter}{p}
\newcommand{\mletter}{m}
\newcommand{\halfCell}[1]{{#1}\minisep}
\newcommand{\encodeCell}[2]{\halfCell{#1}{#2}\sepcell}

\newcommand{\getelem}[2]{{#1}[{#2}]}
\newcommand{\largeAlphabet}{\Gamma}

\newcommand{\inputword}{t}
\newcommand{\polynomial}{P}

\newcommand{\mostgc}[2]{{#1}^{#2}}
\newcommand{\traces}[2]{Traces(\mostgc{#1}{#2})}
\newcommand{\libcfg}{\mu}
\newcommand{\uncalling}{\bot}
\newcommand{\modifyfun}[3]{{#1}[{#2} \leftarrow {#3}]}
\newcommand{\callsymb}{{\tt call}}
\newcommand{\retsymb}{{\tt ret}}
\newcommand{\callevent}[2]{\callsymb({#1},\meth_{#2})}
\newcommand{\retevent}[1]{\retsymb({#1})}
\newcommand{\fullcfg}{\gamma}
\newcommand{\sharedvalue}{d}
\newcommand{\tr}{h}
\newcommand{\mev}{e}

\newcommand{\cfg}{\sf cfg}

\newcommand{\nmethod}{m}
\newcommand{\meth}{M}
\newcommand{\method}{method}
\newcommand{\mstate}{q}
\newcommand{\mStates}{Q}
\newcommand{\mtransitions}{\delta}
\newcommand{\mistate}{\mstate_0}
\newcommand{\mfstate}{\mstate_f}

\newcommand{\domain}{\mathbb{D}}

\newcommand{\spec}{S}
\newcommand{\libr}{Lib}
\newcommand{\stopmeth}{\meth_{\sf Tick}}

\newcommand{\stopval}{{\sf Begin}}
\newcommand{\goval}{{\sf Run}}
\newcommand{\lastval}{{\sf End}}

\newcommand{\readcmd}{{\sf read}}
\newcommand{\writecmd}{{\sf write}}

\newcommand{\ithread}{i}
\newcommand{\imethod}{j}

\newcommand{\nthreads}{k}
\newcommand{\execlength}{l}
\newcommand{\wordsize}{m}

%% file: abstract.tex
\begin{abstract}
It was shown in \citet{journals/iandc/AlurMP00} that the problem of verifying 
finite concurrent systems through \Linearizability{} is in $\EXPSPACE$. However, 
there was still a complexity gap between the easy to obtain $\PSPACE$ lower 
bound and the $\EXPSPACE$ upper bound. We show in this paper that 
\Linearizability{} is $\EXPSPACE$-complete.
\end{abstract}

%% file: introduction.tex
\section{Introduction}

Linearizability~\citep{journals/toplas/HerlihyW90} is the standard consistency 
criterion for concurrent data-structures. 
\citet{journals/tcs/FilipovicORY10}
proved that checking that a library $L$ is linearizable
with respect to a specification $S$ is equivalent to observational refinement.
Formally, as long as linearizability holds, any multi-threaded program using
the specification $S$ as a library can safely replace it by $L$, without adding 
any unwanted behaviors.

Many practical tools
\citep{conf/pldi/BurckhardtDMT10,conf/cav/Vafeiadis10,conf/tacas/ElmasQSST10,conf/spin/VechevYY09,BEEH15} 
for checking linearizability or detecting 
linearizability violations exist, and here is short summary of the work
done on the complexity.

Checking that a single execution is linearizable is already an $\NP$-complete
problem\citep{journals/siamcomp/GibbonsK97}. Moreover,
\citet{journals/iandc/AlurMP00} showed that the problem of checking 
Linearizability for a finite concurrent libraries used by a finite number
of threads is in $\EXPSPACE$ when the specification is a regular language. 
The best known lower bound is 
$\PSPACE$-hardness, obtained from a simple reduction of the reachability 
problem for finite concurrent programs~\citep{journals/iandc/AlurMP00}, 
leaving a large complexity gap.

This result was refined in~\citet{conf/esop/BouajjaniEEH13} where it 
was shown that a simpler variant of Linearizability -- called Static 
Linearizability, or Linearizability with fixed linearization points -- is 
$\PSPACE$-complete for the same class of libraries.

Furthermore, \Linearizability{} is undecidable when the number of threads is 
unbounded~\citep{conf/esop/BouajjaniEEH13}. Tools used for detecting 
linearizability violations often start by underapproximating the set of 
executions by bounding the number of threads. It is thus necessary to 
develop a better understanding of \Linearizability{} for a bounded number of 
threads.

We prove that \Linearizability{} is $\EXPSPACE$-complete, showing that there is
an inherent difficulty to the problem. 
We introduce for this a new problem on regular languages, called 
\problemx{}. This problem can be reduced in polynomial time to Linearizability.

We then show that \problemx{} is $\EXPSPACE$-hard, closing the complexity
gap for Linearizability. Our proof is similar to the proofs of $\EXPSPACE$-
hardness for the problems of inclusion of extended regular expressions with 
intersection operator, or  interleaving operator, given in \citet{H73}, 
\citet{F80} and \citet{MS94}. They all use a similar encoding of runs of Turing 
machines as words, and using the problem at hand, \problemx{} in this case,
to recognize erroneous runs.

To summarize, our two contributions are:
\begin{itemize}
\item 
  finding the \problemx{} problem, a problem equivalent to Linearizability,
  but which has a very simple formulation in terms of regular automata,
\item
  using this problem to show $\EXPSPACE$-hardess of \Linearizability{}.
\end{itemize}

We recall in \sect{sec:definitions} the definition of \Linearizability, and we
introduce the \problemx{} problem. We show in 
\sect{sec:linearizability} that \problemx{} can be reduced in polynomial time
to \Linearizability{}. And finally, we show in \sect{sec:problemx}, that 
\problemx{} is $\EXPSPACE$-hard, which is the most technical part of the paper.
When combined, Sections~\ref{sec:linearizability} and \ref{sec:problemx} show
that \Linearizability{} is $\EXPSPACE$-hard.

%% file: definitions.tex
\section{Definitions}
\label{sec:definitions}

\subsection{Libraries}

In the usual sense, a \emph{library} is a collection of \emph{\method{s}} that 
can be called by other programs. We start by giving our formalism for methods,
and define libraries as sets of \method{s}.

In order to simplify the presentation, and since they do not affect our 
$\EXPSPACE$-hardness reduction, we will use a number of restrictions on 
the \method{s}. First, we will define the \method{s} without 
return values and parameters. Second, each instruction of a \method{} can 
either read or write to the shared memory, but we don't formalize atomic 
compare and set operations. Finally, we limit ourselves to a unique 
\emph{shared variable}.

Let $\domain$ be a finite set used as the \emph{domain} for the  
shared variable and let $\sharedvalue_0 \in \domain$ be a special value
considered as initial.

A \emph{\method} is a tuple $(\mStates,\mtransitions,\mistate,\mfstate)$ where
\begin{itemize}
\item $\mStates$ is the set of states,
\item 
  $\mtransitions \subseteq 
    \mStates \times \set{\readcmd,\writecmd} \times \domain \times \mStates$
\item 
  $\mistate \in \mStates$ is the initial state (in which the method is called)
\item 
  $\mfstate \in \mStates$ is the final state (in which the method can return)
\end{itemize}

One point which might be considered unsual in our formalism is that a 
$\readcmd$ instruction \emph{guesses} the value that it is going to read.
In usual programming languages, this can be understood as first reading a 
variable, and then having an {\sf assume} statement to constrain the value 
of the read variable. This formalism choice is a presentation choice, and has 
no effect on the complexity of the problem.

As hinted previously, a \emph{library} 
$\libr = \set{\meth_1,\dots,\meth_\nmethod}$ is a set of 
methods. For every $\imethod \in \set{1,\dots,\nmethod}$, let 
$(\mStates^\imethod,\mtransitions^\imethod,
\mistate^\imethod,\mfstate^\imethod)$ be the tuple corresponding
to $\meth_\imethod$. We define $\mStates$ to be the (disjoint) union of all 
$\mStates^\imethod$.

Let $\nthreads$ be an integer representing the number of \emph{threads}
using $\libr$. 
Threads run concurrently and call the methods of $\libr$ 
arbitrarily. The system composed of $\nthreads$ threads calling arbitrarily 
the methods of $\libr$ is called $\mostgc{\libr}{\nthreads}$.

Formally, a \emph{configuration} of $\mostgc{\libr}{\nthreads}$ is a pair
$\fullcfg = (\sharedvalue,\libcfg)$ where $\sharedvalue \in \domain$ is the 
current value of the shared variable and $\libcfg$ is a map from 
$\set{1,\dots,\nthreads}$ to $\mStates \uplus \set{\uncalling}$,
specifying, for each thread $\ithread$, the state in which the method called by 
thread $\ithread$ is. 
The symbol $\uncalling$ is used for threads which are \emph{idle} 
(not calling any method at the moment).
 
A \emph{step} from a configuration $\fullcfg = (\sharedvalue,\libcfg)$ to
$\fullcfg' = (\sharedvalue',\libcfg')$ can be:
\begin{itemize}
\item thread $\ithread$ calling method $\imethod$,
denoted by $\fullcfg \xrightarrow{\callevent{\ithread}{\imethod}} \fullcfg'$,
with 
$\libcfg(\ithread) = \uncalling$,
$\libcfg' = \modifyfun{\libcfg}{\ithread}{\mistate^\imethod}$, and
$\sharedvalue' = \sharedvalue$,
\item thread $\ithread$ returning from method $\imethod$,
denoted by $\fullcfg \xrightarrow{\retevent{\ithread}} \fullcfg'$,
with
$\libcfg(\ithread) = \mfstate^\imethod$,
$\libcfg' = \modifyfun{\libcfg}{\ithread}{\uncalling}$, and
$\sharedvalue' = \sharedvalue$,
\item thread $\ithread$ doing a read in method $\imethod$, 
denoted by $\fullcfg \xrightarrow{} \fullcfg'$ (no label)
with
$\libcfg(\ithread) = \mstate \in \meth_\imethod$,
$\libcfg' = \modifyfun{\libcfg}{\ithread}{\mstate'}$,
$(\mstate,\readcmd,\sharedvalue,\mstate') \in \mtransitions^\imethod$,
and $\sharedvalue' = \sharedvalue$,
\item thread $\ithread$ doing a write in method $\imethod$, 
denoted by $\fullcfg \xrightarrow{} \fullcfg'$ (no label)
with
$\libcfg(\ithread) = \mstate \in \meth_\imethod$,
$\libcfg' = \modifyfun{\libcfg}{\ithread}{\mstate'}$,
$(\mstate,\writecmd,\sharedvalue',\mstate') \in \mtransitions^\imethod$.
\end{itemize}

An \emph{execution} of $\mostgc{\libr}{\nthreads}$ is a sequence of 
steps 
$\fullcfg_0 \xrightarrow{} \fullcfg_1 \dots \xrightarrow{} 
\fullcfg_\execlength$
where $\fullcfg_0 = (\sharedvalue_0,\libcfg_0)$, 
with $\libcfg_0(\ithread) = \uncalling$ for all $\ithread$,
is the initial configuration.

The \emph{trace} $\tr$ of an execution is the sequence of labels 
(call's and return's) of its steps.
The set of traces of $\mostgc{\libr}{\nthreads}$ is denoted by 
$\traces{\libr}{\nthreads}$. Note that in a trace, a call event may be without
a corresponding return event (if the method has not returned yet). In which 
case, the call event is said to be \emph{open}.
A trace with no open calls in called \emph{complete}.

Given a complete trace $\tr$, we define for each pair of matching call and 
return events a \emph{method event}. We say that a method event $\mev_1$ 
\emph{happens before} another method event $\mev_2$ if the return event of
$\mev_1$ is before the call event of $\mev_2$ in $\tr$; this defines a 
\emph{happen-before} relation on the method events. The \emph{label} of a 
method event is the method name corresponding to its call event.

\subsection{Linearizability}

Let $\tr$ be a trace of $\traces{\libr}{\nthreads}$ for some library 
$\libr$ and integer $\nthreads$. A complete trace $\tr'$ is said to 
be a \emph{completion} of $\tr$ if we can remove some (possibly zero) open 
calls from $\tr$, as well as close some others open calls (possibly zero) by
adding return events at the end of $\tr$ in order to obtain $\tr'$.

A \emph{specification} for a library 
$\libr = \set{\meth_1,\dots,\meth_\nmethod}$
is a language of finite words $\spec$ over the alphabet 
$\set{\meth_1,\dots,\meth_\nmethod}$.

\begin{definition}[Linearizability]
A complete trace $\tr$ is said to be \emph{linearizable} with respect to
a specification $\spec$ if there exists a total order on the method events, 
respecting the happen-before order, such that the corresponding sequence of 
labels is a word in $\spec$. A trace $\tr$ is said to be \emph{linearizable} 
with respect to $\spec$ if it has a completion which is linearizable 
(with respect to $\spec$).
\end{definition}

\begin{problem}[\Linearizability]
Input: A library $\libr = \set{\meth_1,\dots,\meth_\nmethod}$, a 
non-deterministic finite automaton (NFA) $\spec$ representing the 
specification, and an integer $\nthreads$ given in unary.

Question: Are all the traces of $\traces{\libr}{\nthreads}$ linearizable \wrt 
$\spec$?
\end{problem}

Note: the size of the input is the size of all the automata appearing in the
input (number of states + number of transitions + size of the alphabet) to 
which we add $\nthreads$.

We give in Figs~\ref{fig:ex1}, \ref{fig:ex2}, and \ref{fig:ex3} some examples
to illustrate \Linearizability. To represente executions, we draw a 
method event as an interval, where the left end of the interval corresponds to 
the call event of the method event, and the right end corresponds to the return 
event. This way, when two method events overlap,
they can be ordered arbitrarily, but when a method event $e_1$ is completely 
before a method event $e_2$, $e_1$ has to be ordered before $e_2$.

Above an interval, we write the name of 
the method corresponding to the method event, and below, we write the (unique)
name of the method event.

For all the examples, we consider the regular 
language $\spec = (\meth_A \meth_B)^*$ as a specification.
\fig{fig:ex1} represents an execution which is linearizable, since its method
events can be ordered as the sequence $e_1 e_2 e_3 e_4$, whose corresponding 
sequence of labels is $\meth_A \meth_B \meth_A \meth_B$.
\fig{fig:ex2} represents an execution which is linearizable, since its method
events can be ordered as the sequence $e_1 e_2 e_3 e_4 e_5 e_6$, whose 
corresponding sequence of labels is 
$\meth_A \meth_B \meth_A \meth_B \meth_A \meth_B$.
\fig{fig:ex3} represents an execution which is similar to
\fig{fig:ex1} but is not linearizable.
 
\newcommand{\makeinterval}[6]{
  \draw[|-|] (#2,#1) -- node (#6) {} (#3,#1) ;
  \node[above] at (#6) {#4}; 
  \node[below] at (#6) {#5}; 
}

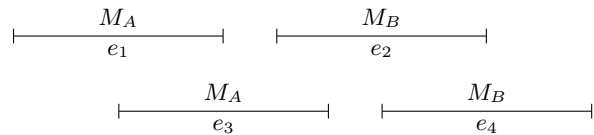
\begin{figure}
\begin{center}
\begin{tikzpicture}[xscale=0.7]
\makeinterval{0}{0}{4}{$\meth_A$}{$e_1$}{a}
\makeinterval{0}{5}{9}{$\meth_B$}{$e_2$}{b}
\makeinterval{-1}{2}{6}{$\meth_A$}{$e_3$}{c}
\makeinterval{-1}{7}{11}{$\meth_B$}{$e_4$}{d}
\end{tikzpicture}
\caption{A linearizable execution, which can be ordered as $e_1 e_2 e_3 e_4$}
\label{fig:ex1}
\end{center}
\end{figure}

\begin{figure}
\begin{center}
\begin{tikzpicture}[xscale=0.5]
\makeinterval{0}{0}{4}{$\meth_A$}{$e_1$}{a}
\makeinterval{0}{5}{9}{$\meth_B$}{$e_2$}{b}
\makeinterval{-1}{2}{6}{$\meth_A$}{$e_3$}{c}
\makeinterval{-1}{7}{11}{$\meth_B$}{$e_4$}{d}
\makeinterval{-2}{3}{8}{$\meth_A$}{$e_5$}{e}
\makeinterval{-2}{10}{13}{$\meth_B$}{$e_6$}{f}
\end{tikzpicture}
\caption{A linearizable execution, which can be ordered as $e_1 e_2 e_3 e_4 e_5 e_6$}
\label{fig:ex2}
\end{center}
\end{figure}

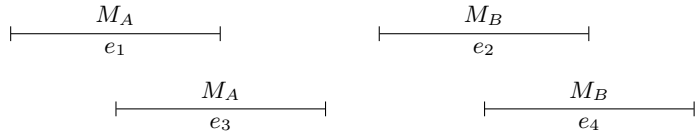
\begin{figure}
\begin{center}
\begin{tikzpicture}[xscale=0.7]
\makeinterval{0}{0}{4}{$\meth_A$}{$e_1$}{a}
\makeinterval{0}{7}{11}{$\meth_B$}{$e_2$}{b}
\makeinterval{-1}{2}{6}{$\meth_A$}{$e_3$}{c}
\makeinterval{-1}{9}{13}{$\meth_B$}{$e_4$}{d}
\end{tikzpicture}
\caption{A non-linearizable execution}
\label{fig:ex3}
\end{center}
\end{figure}

\subsection{\problemx}

We were able to define a new problem, \emph{\problemx}, which: 
1) can be reduced to Linearizability, 
2) is very easy to state (compared to Linearizability), 
3) is still complex enough to capture the difficult part of Linearizability
as we'll show it is $\EXPSPACE$-hard.

\begin{problem}[\problemx]
\label{pb:problemx}
Input: A set of \emph{\insertable} letters 
$\Letters = \set{\letter_1,\dots,\letter_\nletter}$. An NFA $\aut$ over an 
alphabet $\largeAlphabet \uplus \Letters$.

Question: For all words $\word \in \largeAlphabet^*$, does there exist a 
decomposition
$\word = \word_0\cdots\word_\nletter$, and a permutation $\perm$ of 
$\set{1,\dots,\nletter}$, such that 
$\word_0\letter_{\getelem{\perm}{1}}
\word_1\dots\letter_{\getelem{\perm}{\nletter}}\word_\nletter$
is accepted by $\aut$?
\end{problem}

Said differently, for any word of $\largeAlphabet^*$, can we insert the 
letters $\set{\letter_1,\dots,\letter_\nletter}$ 
(each of them exactly once, in any order, anywhere in the word) to obtain a 
word accepted by $\aut$?

Note: the size of the input is the size of N, to which we add $\nletter$.

%% file: linearizability.tex
\section{Reduction from \problemx{} to \Linearizability{}}
\label{sec:linearizability}

In this section, we show that \problemx{} can be reduced in polynomial time 
to \Linearizability{}. When we later show that \problemx{} is $\EXPSPACE$-hard, 
we will get that \Linearizability{} is $\EXPSPACE$-hard as well. 

Intuitively, the letters
$\Letters = \set{\letter_1,\dots,\letter_\nletter}$ of \problemx{} represent 
methods which are all overlapping with every other method, and the word 
$\word$ represents methods which are in sequence. 
\problemx{} asks whether we can insert the letters in $\word$ in order to 
obtain a sequence of $\aut$ while linearizability asks whether there is a
way to order all the letters, while preserving the order of $\word$, 
to obtain a sequence of $\aut$, which is equivalent.

\begin{lemma}
\label{lem:reduction}
\problemx{} can be reduced in polynomial time to \Linearizability{}.
\end{lemma}

\begin{proof}
Let $\Letters = \set{\letter_1,\dots,\letter_\nletter}$ and
$\aut$ an NFA over some alphabet $\Letters \uplus \largeAlphabet$. 

Define $\nthreads$, the number of threads, to be $\nletter + 2$.

We will define a library $\libr$ composed of
\begin{itemize} 
\item methods $\meth_1,\dots,\meth_\nletter$,
one for each letter of $\Letters$
\item methods $\meth_\gamma$, one for 
each letter of $\largeAlphabet$
\item a method $\stopmeth$.
\end{itemize}
and a specification $\spec_\aut$, such that
$(\Letters,\aut)$ is a valid instance of \problemx{} if and only if
$\mostgc{\libr}{\nthreads}$ is linearizable with respect to $\spec_\aut$. 

For the domain of the shared variable, we only need three values: 
$\domain = \set{\stopval,\goval,\lastval}$ with $\stopval$ being the initial
value.

The methods $\meth_\gamma$ are all identical. They just read the value $\goval$ 
from the shared variable (see \fig{fig:methgamma}).

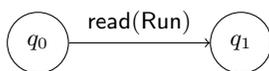
\begin{figure}[htb]
\begin{center}
\begin{tikzpicture}[node distance=2.7cm,->]
\node[state] (A) {$q_0$};
\node[state] (B) [right of=A] {$q_1$};

\path (A) edge node[above] {$\readcmd (\goval)$}  (B);
\end{tikzpicture}
\caption{Description of $\meth_\gamma$, $\gamma \in \largeAlphabet$}
\label{fig:methgamma}
\end{center}
\end{figure}

The methods $\meth_1,\dots,\meth_\nletter$ all read $\stopval$, and then read 
$\lastval$ (see \fig{fig:methi}).

\begin{figure}[htb]
\begin{center}
\begin{tikzpicture}[node distance=2.8cm,->]
\node[state] (A) {$q_0$};
\node[state] (B) [right of=A] {$q_1$};
\node[state] (C) [right of=B] {$q_2$};

\path (A) edge node[above] {$\readcmd (\stopval)$}  (B);
\path (B) edge node[above] {$\readcmd (\lastval)$}  (C);
\end{tikzpicture}
\caption{Description of $\meth_1,\dots,\meth_\nletter$ }
\label{fig:methi}
\end{center}
\end{figure}
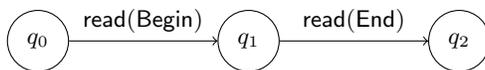

The method $\stopmeth$ writes $\goval$, and then $\lastval$ 
(see \fig{fig:stopmeth}).

\begin{figure}[htb]
\begin{center}
\begin{tikzpicture}[node distance=2.8cm,->]
\node[state] (A) {$q_0$};
\node[state] (B) [right of=A] {$q_1$};
\node[state] (C) [right of=B] {$q_2$};

\path (A) edge node[above] {$\writecmd (\goval)$}  (B);
\path (B) edge node[above] {$\writecmd (\lastval)$}  (C);
\end{tikzpicture}
\caption{Description of $\stopmeth$ }
\label{fig:stopmeth}
\end{center}
\end{figure}

The specification $\spec_\aut$ is defined as the set of words $\word$ over
the alphabet $\set{\meth_1,\dots,\meth_\nletter} \cup \set{\stopmeth}
\cup \set{\meth_\gamma | \gamma \in \Gamma}$
such that one the following condition holds:
\begin{itemize}
\item $\word$ contains 0 letter $\stopmeth$, or more than 1, or
\item for a letter $\meth_i$, $i \in \set{1,\dots,\nletter}$, $\word$
contains 0 such letter, or more than 1, or
\item when projecting over the letters $\meth_\gamma$, 
$\gamma \in \largeAlphabet$ and 
$\meth_i$, $i \in \set{1,\dots,\nletter}$, $\word$ is in $\aut_\meth$,
where $\aut_\meth$ is $\aut$ where each 
letter $\gamma$ is replaced by the letter $\meth_\gamma$, and where
each letter $\letter_i$ is replaced by the letter $\meth_i$.
\end{itemize}

Since $\aut$ is an NFA, $\spec_\aut$ is also an NFA. Moreover, its size is
polynomial is the size of $\aut$.
We can now show the following equivalence:
\begin{enumerate}
\item
  there exists a word $\word$ in $\largeAlphabet^*$, such that there is no 
  way to insert the letters from $\Letters$ in order to obtain a word accepted 
  by $\aut$ \label{stmt:noinsert}
\item 
  there exists an execution of $\libr$ with $\nthreads$ threads which is not 
  linearizable \wrt $\spec_\aut$
  \label{stmt:notlin}
\end{enumerate}

\newcommand{\notlinrun}{r}

$(\ref{stmt:noinsert}) \implies (\ref{stmt:notlin})$. Let $\word \in 
\largeAlphabet^*$ such that there is no way to insert the letters $\Letters$
in order to obtain a word accepted by $\aut$. We construct an execution of 
$\libr$ following \fig{fig:run}, which is indeed a valid execution.

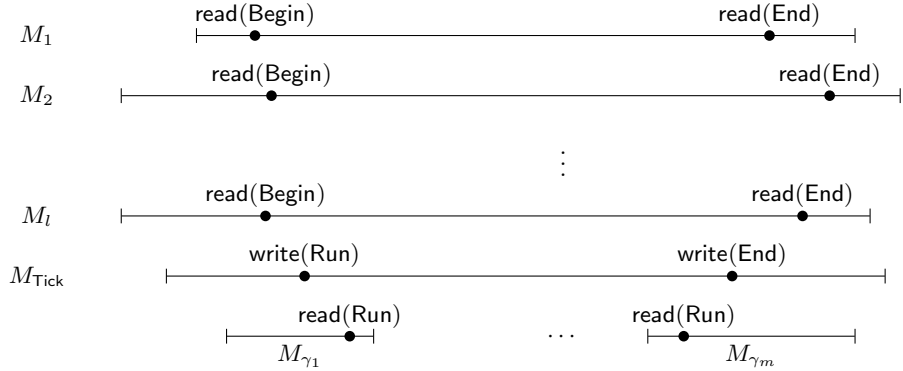
\begin{figure}
\newcommand{\setpoint}[4]{
  \node[circle,fill,inner sep=0pt,minimum size=4pt] at 
    ($(#1) + {#3}*(#2)-{#3}*(#1)$) {};
  \node[above] at ($(#1) + {#3}*(#2)-{#3}*(#1)$) {#4};
}

\begin{tikzpicture}[xscale=2,yscale=0.8]

\node at (0.5,0) (m1call) {};
\node at (5,0) (m1ret) {};
\draw[|-|] (m1call) -- (m1ret);
\node at (-0.5,0) { $\meth_1$ };
\setpoint{m1call}{m1ret}{0.1}{$\readcmd(\stopval)$}
\setpoint{m1call}{m1ret}{0.86}{$\readcmd(\lastval)$}

\node at (0,-1) (m2call) {};
\node at (5.3,-1) (m2ret) {};
\draw[|-|] (m2call) -- (m2ret);
\node at (-0.5,-1) { $\meth_2$ };
\setpoint{m2call}{m2ret}{0.2}{$\readcmd(\stopval)$}
\setpoint{m2call}{m2ret}{0.9}{$\readcmd(\lastval)$}

\node at (3,-2) {$\vdots$};

\node at (0,-3) (mlcall) {};
\node at (5.1,-3) (mlret) {};
\draw[|-|] (mlcall) -- (mlret);
\node at (-0.5,-3) { $\meth_\nletter$ };
\setpoint{mlcall}{mlret}{0.2}{$\readcmd(\stopval)$}
\setpoint{mlcall}{mlret}{0.9}{$\readcmd(\lastval)$}

\node at (0.3,-4) (mtickcall) {};
\node at (5.2,-4) (mtickret) {};
\draw[|-|] (mtickcall) -- (mtickret);
\node at (-0.5,-4) { $\stopmeth$ };
\setpoint{mtickcall}{mtickret}{0.2}{$\writecmd(\goval)$}
\setpoint{mtickcall}{mtickret}{0.78}{$\writecmd(\lastval)$}

\node at (0.7,-5) (mgamma1call) {};
\node at (1.8,-5) (mgamma1ret) {};
\draw[|-|] (mgamma1call) -- node[below] {$\meth_{\gamma_1}$} (mgamma1ret);
\setpoint{mgamma1call}{mgamma1ret}{0.8}{$\readcmd(\goval)$}

\node at (3,-5) {$\cdots$};

\node at (3.5,-5) (mgammakcall) {};
\node at (5,-5) (mgammakret) {};
\draw[|-|] (mgammakcall) -- node[below] 
  {$\meth_{\gamma_\wordsize}$} (mgammakret);
\setpoint{mgammakcall}{mgammakret}{0.2}{$\readcmd(\goval)$};

\end{tikzpicture}
\caption{
  Non-linearizable execution corresponding to a word $\gamma_1\dots\gamma_\wordsize$
  in which we cannot insert the letters from 
  $\Letters = \set{\letter_1,\dots,\letter_\nletter}$ to make it accepted by 
  $\aut$. The points represent steps in the automata.
}
\label{fig:run}
\end{figure}

This execution is not linearizable since
\begin{itemize}
\item it has exactly one $\stopmeth$ method, and
\item for each $i \in \set{1,\dots,\nletter}$, it has exactly one $\meth_i$ 
method, and
\item no linearization of this execution can be in $\aut_\meth$, since
there is no way to insert the letters $\Letters$ into $\word$ to be accepted
by $\aut$.
\end{itemize}

Note: The value of the shared variable is initialized to $\stopval$,
allowing the methods $\meth_i$ ($i \in \set{1,\dots,\nletter})$ to make their 
first transition. $\stopmeth$ then sets the value to $\goval$, thus allowing
the methods $\meth_\gamma, \gamma \in \largeAlphabet$ to execute. Finally,
$\stopmeth$ sets the value to $\lastval$, allowing the methods 
$\meth_\gamma, \gamma \in \largeAlphabet$ to make their second transition and 
return. This tight interaction will enable us to show in the second part of
the proof that all non-linearizable executions of this library have this very
particular form.

$(\ref{stmt:notlin}) \implies (\ref{stmt:noinsert})$. Let 
$\notlinrun$ be an execution which is not linearizable \wrt $\spec_\aut$. We 
first show that this execution should roughly be of the form shown in 
\fig{fig:run}. First, since it is not linearizable \wrt $\spec_\aut$,
it must have at least one completed $\stopmeth$ method event. It it only had
open $\stopmeth$ events (or no $\stopmeth$ events at all), it could be 
linearized by dropping all the open calls to $\stopmeth$. Moreover, it cannot
have more than $\stopmeth$ method event (completed or open), as it could also
linearized, since $\spec_\aut$ accepts all words with more than one 
$\stopmeth$ letter.

We can show similarly that 
for each $i \in \set{1,\dots,\nletter}$, it has exactly one $\meth_i$ 
method which is completed (and none open).

Moreover, the methods $\meth_i$ ($i \in \set{1,\dots,\nletter}$
can only start when the value of the shared variable is $\stopval$, and 
they can only return after reading the value $\lastval$. Since this 
value can only be changed (once) by the single $\stopmeth$ method of our 
executions, 
this ensures that the methods
$\meth_i$ ($i \in \set{1,\dots,\nletter}$) (and $\stopmeth$ itself) all 
overlap with one another, and with every other completed method.

This implies that the completed methods 
$\meth_\gamma, \gamma \in \largeAlphabet$
can only appear in a single thread $t$ 
(since $\meth_1,\dots,\meth_\nletter,\stopmeth$ already occupy $\nletter+1$ 
threads amongst the $\nletter+2$ available).
Thus, we define $\word \in \largeAlphabet^*$ to be the word corresponding to the
completed methods $\meth_\gamma$, $\gamma \in \largeAlphabet$ of the 
execution in the order in which they appear in thread $t$.

Since $\notlinrun$ is not linearizable, we cannot insert 
$\meth_i$ ($i \in \set{1,\dots,\nletter}$) into the completed methods of 
thread $t$ order to be accepted by $\spec_\aut$.
In particular, this implies that there is no way to insert the letters 
$\Letters$ in $\word$ in order to be accepted by $\aut$.

\end{proof}

%% file: proof.tex
\section{\problemx{} is $\EXPSPACE$-hard}
\label{sec:problemx}

We now reduce, in polynomial time, arbitrary exponentially bounded \tmachine{s}, 
to the \problemx{} problem, which shows it is $\EXPSPACE$-hard. We first give a 
few notations.

A \emph{deterministic \tmachine{}} $\tm$ is a tuple 
$(\States,\transitions,\istate,\fstate)$ where:
\begin{itemize}
\item $\States$ is the set of states,
\item 
  $\transitions: (\States \times \set{0,1}) \rightarrow 
   (\States \times \set{0,1} \times \set{\lefttr,\righttr})$
   is the transition function
\item $\istate,\fstate$ are the initial and final states, respectively.
\end{itemize}

A computation of $\tm$ is said to be \emph{accepting} if it ends in $\fstate$.

For the rest of the paper, we fix a \tmachine{} $\tm$ and a polynomial
$\polynomial$ such that all runs of $\tm$ starting with an input of size 
$\tminputsize$ use at most $\ncells$ cells, and such that
the following problem is $\EXPSPACE$-complete.

\begin{problem}[Reachability]
\label{pb:reach}
Input: A finite word $\inputword$.

Question:
Is the computation of $\tm$ starting in state $\istate$, with the tape 
initialized with $\inputword$, accepting?
\end{problem}

\begin{lemma}[\problemx]
\label{lem:problemx}
\problemx{} is $\EXPSPACE$-hard.
\end{lemma}

Note: the sublemmas \ref{lem:notwellformed},\ref{lem:notwellformed2},\ref{lem:insertionerror},\ref{lem:notdelta},\ref{lem:combination}
 are all part of the proof of \lem{lem:problemx}.
\begin{proof}

We reduce in polynomial time the Reachability problem for $\EXPSPACE$
\tmachine{s} to the \emph{negation} of \problemx{}. This still shows that 
\problemx{} is $\EXPSPACE$-hard, as the $\EXPSPACE$ complexity class is closed
under complement.

Let $\tminput$ be a word of size $\tminputsize$.
Our goal is to define a set of letters $\Letters$ and 
an NFA $\aut$ over an alphabet $\largeAlphabet \uplus \Letters$, such
that the following two statements are equivalent:
\begin{itemize}
\item the run of $\tm$ starting in state $\istate$ with the tape initialized 
with $\tminput$ is accepting 
(which, by definition of $\tm$, uses at most $\ncells$ cells),
\item there exists a word $\word$ in $\largeAlphabet^*$, such that there is no 
way to insert (see Problem~\ref{pb:problemx}) the letters $\Letters$ in order 
to obtain a word accepted by $\aut$.
\end{itemize}

More specifically, we will encode runs of our \tmachine{} as words,
and the automaton $\aut$, with the additional set of \insertable{} letters 
$\Letters$, will be used in order to detect words which:
\begin{itemize}
\item don't represent \emph{\wellformedseq} (defined below),
\item 
  or represent a \seqconf{} where the initial configuration is not initialized 
  with $\tminput$ and state $\istate$, or where the final configuration 
  isn't in state $\fstate$,
\item 
  or contain an error in the computation, according to the transition rules of 
  $\tm$.
\end{itemize}

A \emph{configuration} of $\tm$ is an ordered sequence
$(\cell_0,\dots,(\state,\cell_i),\dots,\cell_{\ncells-1})$ representing 
that the content of the tape is $\cell_0,\dots,\cell_{\ncells-1} \in \set{0,1}$, 
the current control state is $\state \in \States$, and the head is on cell $i$.

We denote by $\theint{i}$ the binary representation of $0 \leq i < \ncells$ 
using $\thelog$ digits.
Given a configuration, we represent cell $i$ by:
``$\encodeCell{\theint{i}}{\cell_i}$'' if the head of $\tm$ is not on cell $i$, 
and by ``$\encodeCell{\theint{i}}{\state\cell_i}$'' if the head is on cell $i$
and the current state of $\tm$ is $\state$.
The configuration given above is represented by the word:
\[
  \callconfig
  \encodeCell{\theint{0}}{\cell_0}
  \encodeCell{\theint{1}}{\cell_1}\dots
  \encodeCell{\theint{i}}{\state\cell_i}\dots
  \encodeCell{\theint{\ncells-1}}{\cell_{\ncells-1}}
  \retconfig
\].

Words which are of this form for some 
$\cell_0,\dots,\cell_{\ncells-1} \in \set{0,1}$, $\state \in \States$,
are called \emph{\wellformedcfg{s}}.
A \seqconf{} is then encoded as 
$\callrun\cfg_1\dots \cfg_k \retrun$
where each $\cfg_i$ is a \wellformedcfg. A word of this form is called 
a \emph{\wellformedseq}. We now fix $\largeAlphabet$ to be
$\set{0,1,\callrun,\retrun,\callconfig,\retconfig,\sepcell,\minisep}$.

\begin{lemma}
\label{lem:notwellformed}
There exists an NFA $\notwellformed$ of size polynomial in $\tminputsize$, 
which recognizes words which are not \wellformedcfg{s}.
\end{lemma}

\begin{proof}
A word is not a \wellformedcfg{} if and only if one of the following holds:
\begin{itemize}
\item it is not of the form 
$\callconfig
(\encodeCell{(0+1)^\thelog}{(\States+\epsilon)(0+1)})^*
\retconfig$, or
\item it has no symbol from $\States$, or more than one, or
\item it doesn't start with $\callconfig\halfCell{\theint{0}}$, or
\item it doesn't end with 
$\encodeCell{\theint{\ncells-1}}{(\States+\epsilon)(0+1)}\retconfig$, or
\item it contains a pattern
$\encodeCell{\theint{i}}{(\States+\epsilon)(0+1)}\halfCell{\theint{j}}$ where 
$j \neq i+1$.
\end{itemize}

For all violations, we can make an NFA of size polynomial in $\tminputsize$
recognizing them, and then take their union. The most difficult one is 
the last, for which there are detailed constructions in \citet{F80} and 
\citet{MS94}.

We here give a sketch of the construction. Remember that \theint{i} and
\theint{j} are binary representation using $\thelog$ bits. We want an automaton
recognizing the fact that $j \neq i+1$. The automaton guesses the least
significant bit $b$ ($\thelog$ possible choices) which makes the equality
$i+1 = j$
fails, as well as the presence or not of a carry (for the addition $i+1$) at 
that position. 
We denote by $\theint{i}[b]$ the bit $b$ of $\theint{i}$ and likewise for 
$\theint{j}$. Then, the automaton checks that: 1) there is indeed a violation 
at that position (for instance: no carry, $\theint{i}[b] = 0$ and 
$\theint{j}[b] = 1$) and 2) there is carry if and only if all bits less 
significant that $b$ are set to $1$ is $\theint{i}$.
\end{proof}

\begin{lemma}
\label{lem:notwellformed2}
There exists an NFA $\notrun$ of size polynomial in $\tminputsize$, which 
recognizes words which:
\begin{itemize} 
\item are not a \wellformedseq, or where
\item the first configuration is not in state $\istate$, or
\item the first configuration is not initialized with $\tminput$, or
\item the last configuration is not in state $\fstate$.
\end{itemize}
\end{lemma}
\begin{proof}
Non-deterministic union between $\notwellformed$ and simple automata 
recognizing the last three conditions.
\end{proof}

The problem is now in making an NFA which detects violations in the 
computation with respect to the transition rules of $\tm$. Indeed, in our
encoding, the length of one configuration is about $\ncells$, and thus,
violations of the transition rules from one configuration to the next are going
to be separated by about $\ncells$ characters in the word. We conclude
that we cannot make directly an automaton of polynomial size which 
recognize such violations.

This is where we use the set of \insertable{} letters $\Letters$. We are
going to define and use it here, in order to detect words which encode a 
sequence of configurations where there is a computation error, according to 
the transition rules of $\tm$.

The set $\Letters$, containing $2\thelog$ new letters, is
defined as 
$\Letters = \ab\set{
  \pletter_1,\dots,\pletter_\thelog,\ab
  \mletter_1,\ab,\dots,\mletter_\thelog
}$.

We want to construct an NFA $\notdelta$, such that, for a word $\word$ which is
a \wellformedseq, these statements are equivalent:
\begin{itemize}
\item 
  $\word$ has a computation error according to the transition rules 
  $\transitions$ of $\tm$ 
\item 
  we can insert the letters 
  $\Letters$ in $\word$ to obtain a word accepted by $\notdelta$.
\end{itemize}

The idea is to use the letters $\Letters$ in order to identify two places in 
the word corresponding to the same cell of $\tm$, but at two successive 
configurations of the run.

As an example, say we want to detect a violation of the transition 
$\transitions(\state,0) = (\state',1,\righttr)$, that is, which 
reads a $0$, writes a $1$, moves the head to the right, and changes the state 
from $\state$ to $\state'$.
  
Assume that $\word$ contains a sub-word of the following form:
\begin{align*}
\encodeCell{\theint{i}}{\state0}
\dots
\callconfig
\dots
\encodeCell{\theint{i}}{1}
\encodeCell{\theint{i+1}}{\state''\cell_{i+1}} 
\end{align*}
where $\state''$ is different than $q'$

The single $\callconfig$ symbol on the middle of the sub-word 
ensures that we are checking violations in successive configurations. Here,
with the current state being $\state$, the head read $0$ on cell $i$, wrote $1$
successfully, and moved to the right.
But the state changed to $\state''$ instead of $\state'$. Since we assumed that
$\tm$ is deterministic, this is indeed a violation of the transition rules.

We now have all the ingredients in order to construct $\notdelta$. It will be 
built as a non-deterministic choice (or union) of $\notdeltaexample{t}$
for all possible transitions $t \in \transitions$ (with $\transitions$ 
seen as a relation).

As an example, we show how to construct the automaton $\notdeltaex$,
part of $\notdelta$, and recognizing violations of 
$\transitions(\state,0) = (\state',1,\righttr)$, where the head was indeed 
moved to right, but the state was changed to some state $\state''$ instead of 
$\state'$, like above. Other violations may be recognized similarly.

$\notdeltaex$ starts by finding a sub-word of the form (the $+$ denotes the
disjunction or union of regular expressions, and $^*$ denotes the Kleene star,
0 or more repetitions):
\begin{flalign}
\label{eq:firstpermutation}
&\encodeCell{
  (\mletter_10+\pletter_11)\dots
  (\mletter_\thelog 0+\pletter_\thelog 1)
}{\state0}
\end{flalign}
meaning the state is $\state$ and the head points to a cell containing $0$.
After that, it reads arbitrarily many symbols, but exactly one $\callconfig$
symbol, which ensures that the next letters it reads are from the next 
configuration. Finally, it looks for a sub-word of the form
\begin{flalign}
\label{eq:secondpermutation}
&\encodeCell{
  (\pletter_10+\mletter_11)\dots
  (\pletter_\thelog 0+\mletter_\thelog 1)
}{(0+1)}
\encodeCell{(0+1)^*}\state''
\end{flalign}
for some $\state'' \neq \state'$.

We can now show the following.
\begin{lemma}
\label{lem:insertionerror}
For a \wellformedseq{} $\word$, these two statements are equivalent:
\begin{enumerate}
\item there is a way to insert the letters $\Letters$ into $\word$ to be 
accepted by $\notdeltaex$ \label{stmt:insertion}
\item in the sequence of configurations encoded by $\word$, there is a 
configuration where the state was $\state$ and the head was pointing to a cell 
containing $0$, and in the next configuration, the head was moved to the right,
but the state was not changed to $\state'$ (computation error). 
\label{stmt:error}
\end{enumerate}
\end{lemma}

\begin{proof}

$(\Leftarrow)$.
We 
insert the letters $\Letters$ in front of the binary representation of the 
cell number where the violation occurs. The violation involves two 
configurations: in the first, we insert $\mletter$'s in front of $0$'s, and 
$\pletter$'s in front of $1$'s, and in the second, it's the other way around.

This way, we inserted all the letters of $\Letters$ (exactly) once into 
$\word$, and $\notdeltaex$ is now able to recognize the 
patterns~(\ref{eq:firstpermutation}) and (\ref{eq:secondpermutation})
described above.

$(\Rightarrow)$.
For the other direction, let $\word$ be a \wellformedseq{} such that there 
exists a way to insert the letters $\Letters$ into $\word$, in order to obtain 
a word $\word_\Letters$ accepted by $\notdeltaex$.

Since each letter of $\Letters$ can be inserted only once, the sub-word 
matched by
$(\mletter_10+\pletter_11)\dots(\mletter_\thelog 0+\pletter_\thelog 1)$
in pattern~(\ref{eq:firstpermutation}) in $\notdeltaex$ has to be the
same as the one matched by 
$(\pletter_10+\mletter_11)\dots(\pletter_\thelog 0+\mletter_\thelog 1)$ 
in pattern~(\ref{eq:secondpermutation}), up to 
exchanging $\mletter$'s and $\pletter'$'s.

Moreover, having exactly one $\callconfig$ symbol in between the two patterns
ensures that they correspond to the same cell, but in two successive 
configurations.

Finally, the facts that $\state''$ is different that $\state'$ and that 
$\tm$ is deterministic ensures that the sequence of configurations represented 
by $\word$ indeed contains a computation error according to the rule 
$\transitions(\state,0) = (\state',1,\righttr)$.
\end{proof}

We thus get the following lemma for the automaton $\notdelta$.

\begin{lemma}
\label{lem:notdelta}
For a word $\word$ which is \wellformedseq, these statements are equivalent:
\begin{itemize}
\item 
  we can insert the letters 
  $\Letters$ in $\word$ to obtain a word accepted by $\notdelta$,
\item 
  $\word$ has a computation error according to the transition rules 
  $\transitions$ of $\tm$.
\end{itemize}
\end{lemma}

\begin{proof}
Construct all the $\notdeltaexample{t}$ for $t \in \transitions$ 
(with $\transitions$ considered as a relation). Construct similarly an
automaton recognizing the violation where a cell changes while the head
was not here. Take the union of all these automata, the 
proof then follows from \lem{lem:insertionerror}.
\end{proof}

By taking the union $\aut = \notrun \cup \notdelta$, we finally get the 
intended result, which ends the reduction.

\begin{lemma}
\label{lem:combination}
The following two statements are equivalent.
\begin{itemize}
\item the run of $\tm$ starting in state $\istate$ with the tape initialized 
with $\tminput$ is accepting,
\item there exists a word $\word$ in $\largeAlphabet^*$, such that there is no 
way to insert the letters $\Letters$ in order to obtain a 
word accepted by $\aut$.
\end{itemize}
\end{lemma}
\begin{proof}
$(\Rightarrow)$ Let $\word$ be the \wellformedseq{} representing the sequence
of configurations of the accepting run in $\tm$, with the tape initialized
with $\tminput$. Then by \lem{lem:notwellformed2} and \lem{lem:notdelta},
there is no way to insert the letters $\Letters$ in order to obtain a word
accepted by $\notrun$ or $\notdelta$.

$(\Leftarrow)$ Let $\word \in \largeAlphabet^*$ be a word such that there is no
way to insert the letters $\Letters$ in order to obtain a word accepted by 
$\aut$. First, since $\word$ is not accepted by $\notrun$, it represents
a \wellformedseq{}, starting in state $\istate$ with the tape initialized
with $\tminput$ and ending in state $\fstate$ (\lem{lem:notwellformed2}).
Moreover, since there is no way to insert the letters to obtain a word
from $\notdelta$, $\word$ has no computation error according to the transition
rules $\transitions$ of $\tm$ (\lem{lem:notdelta}).
\end{proof}

This ends the proof of \lem{lem:problemx}.
\end{proof}

Since \problemx{} is $\EXPSPACE$-hard and, \problemx{} reduces to 
\Linearizability{}, we get the main result of the paper.
 
\begin{theorem}[\Linearizability]
Linearizability{} is $\EXPSPACE$-complete.
\end{theorem}

\begin{proof}
It was previously shown that \Linearizability{} is in $\EXPSPACE$~\citep{journals/iandc/AlurMP00}.
$\EXPSPACE$-hardness follows from 
Lemmas~\ref{lem:reduction} and \ref{lem:problemx}
\end{proof}

%% file: conclusion.tex
\section{Conclusion}

We define a new problem, \problemx{}, simpler than \Linearizability{}, but 
still hard enough to capture the main difficulties of \Linearizability{}. 
We showed that the \problemx{} problem is $\EXPSPACE$-hard, and could thus 
deduce that the \Linearizability{} problem is $\EXPSPACE$-hard. 

Our 
result applies even with all
the following restrictions:
the number of threads is given in unary,
there is a unique shared variable whose domain size is $3$,
  the library has a constant number of automata ``shapes'' 
  (3 in our reduction) using less than $3$ states,
the methods of the library are deterministic,
the methods of the library have no loop,
and the instructions within the methods can only read or write, but never do 
both atomically.

For future work, we plan to show that restricting ourselves to deterministic 
specifications (using a DFA instead of an NFA in the input of the problem) does 
not reduce the complexity.
Furthermore, it would be interesting to find a large class of specifications 
including the most common ones 
(stack, queue, \dots) for which our lower-bound does not apply and where we
could reduce the complexity.